\documentclass[11pt]{article}

\usepackage{amsmath,amssymb}
\usepackage[all]{xy}

\font\fr=eufm10 scaled \magstep 1 
\newtheorem{theorem}{Theorem}
\newtheorem{prop}{Proposition}

\newtheorem{lem}{Lemma}
\newtheorem{definition}{Definition}
\newtheorem{remarkth}[definition]{Remark}

\def\qed{\ifvmode\Realemovelastskip\fi
{\unskip\nobreak\hfil\penalty50\hbox{}\nobreak\hfil \hbox{\vrule
height1.2ex width1.2ex}\parfillskip=0pt \finalhyphendemerits=0
\par\smallskip}}

\def\qedr{\ifvmode\Realemovelastskip\fi
{\unskip\nobreak\hfil\penalty50\hbox{}\nobreak\hfil \hbox{
$\diamond$}\parfillskip=0pt \finalhyphendemerits=0
\par\smallskip}}

\def\ds{\displaystyle}

\newenvironment{proof}{\noindent{\sl Proof:~~~}}{\quad \qed}

\def\beq{\begin{equation}}
\def\eeq{\end{equation}}
\def\bea{\begin{eqnarray}}
\def\eea{\end{eqnarray}}
\def\beann{\begin{eqnarray*}}
\def\eeann{\end{eqnarray*}}
\def\beasn{\begin{sneqnarray}}
\def\eeasn{\end{sneqnarray}}
\def\ben{\begin{enumerate}}
\def\een{\end{enumerate}}
\def\bit{\begin{itemize}}
\def\eit{\end{itemize}}
\def\proof{ ({\sl Proof\/}) }

\def\derpar#1#2{\displaystyle\frac{\partial{#1}}{\partial{#2}}}

\def\vf{\mbox{\fr X}}
\def\df{{\mit\Omega}}

\def\d{{\rm d}}

\def\Real{\mathbb{R}}

\def\r{\ensuremath{\mathbb{R}}}
\def\rk{{\mathbb R}^{k}}
\def\rkq{\rk \times Q}
\def\Tan{{\rm T}}
\def\Lie{\mathop{\rm L}\nolimits}
\def\inn{\mathop{i}\nolimits}
\def\Cinfty{{\rm C}^\infty}
\def\longto{\longrightarrow}
\def\tabaddress#1{{\small\it\begin{tabular}[t]{c}#1
\\[1.2ex]\end{tabular}}}
\def\qed{\ifvmode\removelastskip\fi
{\unskip\nobreak\hfil\penalty50\hbox{}\nobreak\hfil \hbox{\vrule
height1.2ex width1.2ex}\parfillskip=0pt \finalhyphendemerits=0
\par\smallskip}}

\title{ON A KIND OF NOETHER SYMMETRIES AND
 CONSERVATION LAWS  IN $k$-COSYMPLECTIC FIELD THEORY}

\author{{\sc Juan Carlos Marrero\thanks{{\bf e}-{\it mail}:jcmarrer@ull.es}}
   \\ \tabaddress{Departamento de Matem\'{a}tica Fundamental
   Facultad de Matem\'{a}ticas\\ Universidad de la Laguna, Spain
   }\\
{\sc Narciso Rom\'an-Roy\thanks{{\bf e}-{\it mail}:
   nrr@ma4.upc.edu}}
   \\
   \tabaddress{Departamento de Matem\'atica Aplicada IV.
   Edificio C-3, Campus Norte UPC\\
   C/ Jordi Girona 1. 08034 Barcelona. Spain}
   \\
{\sc Modesto Salgado\thanks{{\bf e}-{\it mail}:
modesto.salgado@usc.es}}
\\
  \tabaddress{Departamento de Xeometr\'\i a e Topolox\'\i a,
     Facultade de Matem\'aticas \\
     Universidade de Santiago de Compostela.
     15782- Santiago de Compostela. Spain}
     \\
 {\sc Silvia Vilari\~no\thanks{{\bf
e}-{\it mail}:
silvia.vilarino@udc.es}} \\
  \tabaddress{Departamento de Matem\'{a}ticas,
     Facultade de Ciencias. \\
     Universidad de A Coru\~{n}a.
     15071-A Coru\~{n}a. Spain}}

\begin{document}

\maketitle

\pagestyle{myheadings}

\thispagestyle{empty}

\begin{abstract}
This paper is devoted to studying symmetries of
certain kinds of $k$-cosymplectic
Hamiltonian systems in first-order classical field theories.
Thus, we introduce a particular class of symmetries and study
the problem of associating conservation laws to them by means of
a suitable generalization of Noether's theorem.
\end{abstract}


\vbox{\raggedleft AMS s.\,c.\,(2000): 70S05, 70S10, 53D05}\null
\markright{\rm J.C. Marrero {\it et al\/},
   \sl Noether symmetries in $k$-cosymplectic field theory}

 \clearpage


\section{Introduction}
\label{intro}

The {\sl $k$-cosymplectic formalisms}
is one of the simplest geometric frameworks for describing
first-order classical field theories.
It is the generalization to field
theories of the standard cosymplectic formalism for non-autonomous mechanics,
 \cite{mod1,mod2}, and it describes field
theories involving the coordinates in the basis on the Lagrangian
 and on the Hamiltonian. The foundations of the $k$-cosymplectic formalism
are the $k$-cosymplectic manifolds \cite{mod1,mod2}.

Historically, it is based on the so-called {\sl polysymplectic formalism}
developed by G\"{u}nther \cite{gun}, who introduced
{\sl polysymplectic manifolds}. A refinement of this concept
led to define {\sl $k$-symplectic manifolds} \cite{aw1,aw2,aw3}, which are
polysymplectic manifolds admiting {\sl Darboux-type coordinates}
 \cite{mt1}.
(Other different polysymplectic formalisms
for describing field theories have been also proposed
\cite{Sarda2,Kana,McN,No2,No5,Sd-95}).

The natural extension of the $k$-symplectic manifolds are the
$k$-cosymplectic manifolds. All of this is discused in Section \ref{hamil},
which is devoted to make a review on the main features
 and characteristics of $k$-cosymplectic manifolds and of $k$-cosymplectic Hamiltonian systems.
 We also introduce the notion of {\sl almost standard $k$-cosymplectic manifold},
 which are those that we are interested in this paper.

The main objective of this paper is  to study  symmetries and conservation
laws on first-order classical field theories, from the Hamiltonian viewpoint,
 using the $k$-cosymplectic description,
and considering only the regular case. These problems have been
treated for $k$-symplectic field theories in \cite{fam,RSV-09},
generalizing the results obtained for non-autonomous mechanical
systems
 (see, in particular, \cite{LM-96}, and references quoted therein).
 We further remark that the problem of symmetries in field theory
has also been analyzed using other geometric frameworks, such as
the {\sl multisymplectic models} (see, for instance,
\cite{CM-2003,EMR-99,Go2,Gymmsy,Kijo,KijTul,LMS-2004}).

In this way, in Section \ref{ntmvf}
 we recover the idea of {\sl conservation law}
 or {\sl conserved quantity}.
 Then, we introduce a particular kind of symmetries for
 (almost-standard) $k$-cosymplectic Hamiltonian systems,
 essentially those transformations preserving the $k$-cosymplectic structure, which
 allows us to state a generalization of Noether's theorem.
 The definition of these so-called {\sl $k$-cosymplectic Noether symmetries}
 is inspired in the ideas introduced by C. Albert in his study of
symmetries for the cosymplectic formalism of autonomous mechanical systems \cite{albert}.

 Finally, as an example, in Section \ref{ex} we describe briefly
 the k-cosymplectic quadratic Hamiltonian systems and we analyze some
Noether symmetries for these kinds of systems
(in particular, for the wave equation).

 In this paper, manifolds are real, paracompact,
  connected and $C^\infty$, maps are $C^\infty$, and sum over crossed repeated
  indices is understood.

\section{Geometric elements. Hamiltonian $k$-cosymplectic formalism}
\label{hamil}

(The contents of this section can be seen in more detail in
\cite{mod2}).

\subsection{k-vector fields and integral sections}

 Let $M$ be an arbitrary manifold, $T^{1}_{k}M$  the Whitney
sum $TM\oplus \stackrel{k}{\dots} \oplus TM$ of $k$ copies of $TM$
and $\tau : T^{1}_{k}M \longrightarrow M$ its  canonical projection.
$T^{1}_{k}M$ is usually called the tangent
bundle of $k^1$-velocities of $M$.

\begin{definition}  \label{kvector}
A {\rm $k$-vector field} on $M$ is a section
${\bf X} : M \longrightarrow T^1_kM$ of the projection $\tau$.
\end{definition}

Giving a $k$-vector field ${\bf X}$ is equivalent to
giving a family of $k$ vector fields $X_{1}, \dots, X_{k}$ on $M$
obtained by projecting ${\bf X}$ onto every factor; that is,
$X_A=\tau_A\circ{\bf X}$, where $\tau_A\colon T^1_kM \rightarrow
TM$ is the canonical projection onto the $A^{th}$-copy $TM$ of
$T^1_kM$. For this reason we will denote a $k$-vector field by
${\bf X}=(X_1, \ldots, X_k)$.

\begin{definition} \label{integsect} An
{\rm integral section}  of the $k$-vector field  \, $(X_{1}, \dots,
X_{k})$ \, passing through a point $x\in M$  is a map
 $\psi:U_0\subset \r^k \rightarrow M$, defined on some neighborhood
 $U_0$ of $0\in \rk$, such that
$$
\psi(0)=x, \quad
\psi_{*}(t)\left(\displaystyle\frac{\partial}{\partial t^A}\Big\vert_t\right) =
 X_{A}(\psi (t)) \, \quad \mbox{for every} \quad
t\in U_0 \, ,
$$
A $k$-vector field ${\bf X}$ is {\rm integrable} if
every point of~$M$ belongs to the image of an integral section of~${\bf X}$.
\end{definition}

In coordinates, if
$\ds X_A = X_A^i \frac{\partial}{\partial q^i}$,
then $\psi$ is an integral section of $\mathbf{X}$ if, and only if,
the following system of partial differential equations holds:
$$
\frac{\partial \psi^i}{\partial t^A} = X_A^i\circ \psi .
$$

We remark that a $k$-vector field ${\bf X}=(X_1,\ldots , X_k)$ is
integrable if, and only if, the vector fields $X_1, \ldots , X_k$
generate a completely integrable distribution of rank $k$. This is
the geometric expression of the integrability condition of the
preceding differential equation (see, for instance,
\cite{Die,Lee}).

 Observe that, in case $k=1$,
 this definition coincides with
the definition of integral curve of a vector field.

\subsection{$k$-symplectic manifolds}

The  polysymplectic structures were introduced in \cite{gun} and the
$k$-symplectic structures in \cite{aw1,mt1}.
\begin{definition}
\label{defaw}
Let $M$ be a  differentiable manifold of dimension $N=n+kn$.
\ben
\item
A {\rm polysymplectic structure} on $M$
is a family $(\omega_0^A)$ ($1\leq A\leq k$),
where each $\omega_0^A\in\df^2(M)$ is a closed form,
such that
$$
\cap_{A=1}^{k} \ker\omega_0^A = \{0\} .
$$
Then $(M,\omega_0^A)$ is called a {\rm polysymplectic manifold}.

\item
A {\rm $k$-symplectic structure} on $M$ is a family
$(\omega_0^A,V) $ ($1\leq A\leq k$), such that $(M,\omega_0^A)$ is
a polysymplectic manifold and $V$ is an integrable
$nk$-dimensional tangent distribution on $M$ satisfying that,
\[
 \omega_0^A \vert_{V\times V} =0, \; \; \mbox{ for every } A.
\]
Then $(M,\omega_0^A,V)$ is called a {\rm $k$-symplectic manifold}.
\een
The $k$-symplectic (resp., polysymplectic) structure is {\rm
exact} if $\omega_0^A =\d \theta_0^A$, for all $A$.
\end{definition}

\begin{theorem}\label{dtk}
\cite{mt1}.
Let $(\omega_0^A,V)$ be a $k$-symplectic structure on~$M$.
For every point of $M$ there exists a neighbourhood $U$
and local coordinates
$(q^i , p^A_i)$ ($1\leq i\leq n$, $1\leq A\leq k$)
such that, on~$U$,
$$
\omega_0^A=  dq^i\wedge dp^A_i ,
\quad
V =
\left\langle  \frac{\partial}{\partial p^1_i}, \dots,
\frac{\partial}{\partial p^k_i} \right\rangle_{i=1,\ldots,n} .
$$
\end{theorem}
These are called \emph{Darboux} or \emph{canonical coordinates}
of the $k$-symplectic manifold.

The canonical model of a $k$-symplectic manifold is
$((T^1_k)^*Q,\omega_0^A,V)$,
where $Q$ is a $n$-dimensional differentiable manifold
 and $(T^1_k)^*Q= T^*Q \oplus \stackrel{k}{\dots} \oplus
T^*Q$ is the Whitney sum of $k$ copies of the cotangent bundle
$T^*Q$,
which is usually called the {\sl bundle of $k^1$-covelocities} of $Q$
(see \cite{kms}).
We have the natural projections
$$
\begin{array}{ccccccccc}
\pi^A \colon & (\Tan^1_k)^*Q & \rightarrow & \Tan^*Q & ; &
(\pi_Q)_1 \colon & (\Tan^1_k)^*Q & \to & Q
\\
& (q;\alpha^1_q, \ldots ,\alpha^k_q) &\mapsto& (q;\alpha^A_q) &  &
& (q;\alpha^1_q, \ldots ,\alpha^k_q)&\mapsto& q
\end{array} \ .
$$
The manifold
$(T^1_k)^*Q$ can be identified with the manifold $J^1(Q,\rk)_0$ of
1-jets of mappings from $Q$ to $\rk$ with target at $0\in \rk$,
that is
 \[\begin{array}{ccc}
  J^1(Q,\r^k)_0 & \equiv & T^*Q \oplus \stackrel{k}{\dots} \oplus T^*Q \\
 j^1_{q,0}\sigma_{Q}  & \equiv & (d\sigma_{Q}^1(q), \dots ,d\sigma_{Q}^k(q))
\end{array}\]
 where $\sigma_Q^A= \pi^A_0 \circ \sigma_Q:Q \longrightarrow \r$ is the
 $A^{th}$ component of $\sigma_Q$, and  $\pi^A_0:\r^k \to \r$ is the canonical
 projection onto the $A$ component.

$(\Tan^1_k)^*Q$ is endowed with the canonical forms
$$
\theta^A=(\pi^A)^*\theta_0 ,
\quad
\omega_0^A=(\pi^A)^*\omega_0=-(\pi^A)^*\d\theta_0\ ,
$$
where $\theta_0$ and $\omega_0=-\d\theta_0$ are the Liouville $1$-form and
the canonical symplectic form on $\Tan^*Q$.
Obviously $\omega_0^A= -\d \theta_0^A$.

If $(q^i)$ ($1\leq i\leq n$) are local coordinates on $U \subset Q$,
the induced coordinates $(q^i ,p^A_i)$ ($1\leq A\leq k$) on
$(\pi^1_Q)^{-1}(U)$ are given by
$$
q^i(q;\alpha^1_q, \ldots ,\alpha^k_q)= q^i(q) ,
\quad
p^A_i(q;\alpha^1_q, \ldots ,\alpha^k_q)=
\alpha^A_q\left(\frac{\partial}{\partial q^i}\Big\vert_q\right) .
$$
Then we have
$$
\theta_0^A = p^A_i\d q^i ,
\quad
\omega_0^A =
\d q^i\wedge\d p^A_i .
$$
Thus, the triple  $((\Tan^1_k)^*Q,\omega_0^A, V)$, where $V=\ker
\Tan (\pi_Q)_1$, is a $k$-symplectic manifold, and the natural
coordinates in $(\Tan^1_k)^*Q$ are Darboux coordinates.

\subsection{$k$-cosymplectic manifolds}

\begin{definition}
\label{deest}
Let ${\cal M}$ be a a differentiable manifold of dimension $N=k+n+kn$.
\ben
\item
A {\rm polycosymplectic structure} in ${\cal M}$ is a family
$(\eta^A,\omega^A) $ ($1\leq A\leq k$), where
$\eta^A\in\df^1(M)$ and $\omega^A\in\df^2({\cal M})$
are closed forms satisfying that
\begin{enumerate}
\item
$\eta^1\wedge\ldots\wedge\eta^k\not=0$.
\item
$({\cap_{A=1}^{k}}\ker\omega^A{\cap_{A=1}^{k}}\ker\eta^A)=\{0\}$.
\een
Then, $({\cal M},\eta^A,\omega^A)$ is said to be a {\rm polycosymplectic manifold}.
\item
A {\rm $k$--cosymplectic structure}  in ${\cal M}$ is a family
$(\eta^A,\omega^A,{\cal V})$ such that
 $({\cal M},\eta^A,\omega^A)$ is a polycosymplectic manifold,
 and ${\cal V}$ is an $nk$-dimensional integrable distribution on ${\cal M}$,
 satisfying that
\begin{enumerate}
\item
$\eta^A\vert_{\cal V}=0$.
\item
$\omega^A\vert_{{\cal V}\times {\cal V}}=0$.
\end{enumerate}
Then, $({\cal M},\eta^A,\omega^A,{\cal V})$
 is said to be a {\rm $k$--cosymplectic manifold}.
\een
The $k$-cosymplectic (resp., polycosymplectic) structure
is {\rm exact} if $\omega^A =\d \theta^A$, for all $A$.
\end{definition}

For every $k$-cosymplectic structure  $(\eta^A ,\omega^A,{\cal V})$ on ${\cal M}$,
 there exists a family of $k$ vector fields
$\{R_A\}_{\, 1\leq A\leq k}$, which are called {\sl Reeb vector fields},
 characterized by the following conditions \cite{mod1}
$$
\inn(R_A)\eta^B=\delta^B_A \quad ,\quad \inn(R_A)\omega^B=0\quad ;
\quad 1\leq A,B \leq k \ .
$$

\begin{theorem}  {\rm (Darboux Theorem)} \cite{mod1}:
If ${\cal M}$ is a $k$--cosymplectic manifold, then for every point of
$M$ there exists a local chart of coordinates $(t^A,q^i,p^A_i)$,
$1\leq A\leq k$, $1\leq i \leq n$, such that
$$
\eta^A=\d t^A,\quad \omega^A=\d q^i\wedge\d p^A_i,
 \quad {\cal V}=\left\langle \frac{\partial} {\partial p^1_i}, \dots,
 \frac{\partial}{\partial p^k_i}\right\rangle_{i=1,\ldots , n}, \quad
 R_A=\derpar{}{t^A}.
$$
\label{darbo}
\end{theorem}
These are called \emph{Darboux} or \emph{canonical coordinates}
of the $k$-cosymplectic manifold.

The canonical model for $k$-cosymplectic manifolds is
$(\rk\times (T^1_k)^*Q,\eta^A,\omega^A,{\cal V})$.
The manifold $J^1\pi_{Q}$ of  1-jets of sections of the trivial
bundle $\pi_{Q}:\rk \times Q \to Q$ is diffeomorphic to
$\rk \times(T^1_k)^*Q$.
We use also the following notation for the
canonical projections
$$
(\pi_Q)_1\colon\rk \times (T^1_k)^*Q \stackrel{(\pi_{Q})_{1,0}}{\longrightarrow}
\rk \times Q \stackrel{\pi_{Q}}{\longto} Q
$$
given by
$$
\pi_Q(t,q)=q, \quad (\pi_Q)_{1,0}(t,\alpha^1_q, \ldots ,\alpha^k_q)=(t,q), \quad
(\pi_Q)_1(t,\alpha^1_q, \ldots ,\alpha^k_q)=q \, ,
$$
with $t\in\rk $, $q\in Q$ and $(\alpha^1_q, \ldots ,\alpha^k_q)\in(T^1_k)^*Q$.

If $(q^i)$ are local coordinates on $U \subseteq Q$, then the
induced local coordinates  $(t^A,q^i ,p^A_i)$
on $[(\pi_Q)_1]^{-1}(U)=\rk \times (T^1_k)^*U$ are given by
$$
t^A(t,\alpha^1_q, \ldots ,\alpha^k_q) = t^A; \quad
q^i(t,\alpha^1_q, \ldots ,\alpha^k_q) = q^i(q); \quad
 p^A_i(t,\alpha^1_q, \ldots ,\alpha^k_q) =
\alpha^A_q\left(\ds\frac{\partial}{\partial q^i}\Big\vert_q \right)
$$

On $\rk\times (T^1_k)^*Q$, we define the differential forms
$$
\eta^A=(\pi^A_1)^*\d t^A\, , \quad \theta^A= (\pi^A_2)^*\theta_0\, ,
\quad
\omega^A= (\pi^A_2)^*\omega_0\, ,
$$
where $\pi^A_1:\rk \times (T^1_k)^*Q \rightarrow \r$ and
$\pi^A_2:\rk \times (T^1_k)^*Q \rightarrow T^*Q$ are the projections
defined by
$$
\pi^A_1(t,(\alpha^1_q, \ldots ,\alpha^k_q))=t^A
\,,\quad  \pi^A_2(t,(\alpha^1_q, \ldots ,\alpha^k_q))=\alpha^A_q\, ,
$$

In local coordinates we have
$$
  \eta^A=\d t^A \, , \quad \theta^A = \displaystyle \sum_{i=1}^n \, p^A_i
 dq^i \,   , \quad \omega^A = \displaystyle \sum_{i=1}^n dq^i
\wedge dp^A_i\, .
$$

Moreover, let ${\cal V}=ker \, T(\pi_Q)_{1,0}$. Then
$\ds{\cal V}=\left\langle \frac{\partial}{\partial p^1_i}, \dots,
\frac{\partial}{\partial p^k_i}\right\rangle_{i=1,\ldots , n}$.

Hence $(\rk\times (T^1_k)^*Q,\eta^A,\omega^A,{\cal V})$ is a
$k$-cosymplectic manifold, and the natural coordinates of $\rk\times
(T^1_k)^*Q$ are Darboux coordinates for this canonical
$k$-cosymplectic structure. Furthermore, $\ds
\left\{\derpar{}{t^A}\right\}$ are the Reeb vector fields of this
structure.

Now, let $\varphi:\rkq\longto\rkq$ be a
diffeomorphism of $\pi_{Q}$-fiber bundles, and let $\varphi_{Q}:Q\longto Q$ be the
diffeomorphism induced on the base. We can lift $\varphi$ to a
diffeomorphism $j^{1*}\varphi:\rk\times (T^1_k)^*Q\longto\rk\times
(T^1_k)^*Q$ such that the following diagram commutes:
\[\xymatrix{J^1\pi_Q\equiv\rk\times (T^1_k)^*Q
\ar[rr]^{\txt{\small{$j^{1\,*}\varphi$}}}
\ar[dd]_{\txt{\small{$(\pi_{Q})_{1,0}$}}} & &
J^1\pi_Q\equiv\rk\times (T^1_k)^*Q
\ar[dd]^{\txt{\small{$(\pi_{Q})_{1,0}$}}}
\\ \\
\rkq \ar[rr]^{\txt{\small{$\varphi$}}}
\ar[dd]_{\txt{\small{$\pi_{Q}$}}}
& & \rkq
\ar[dd]^{\txt{\small{$\pi_{Q}$}}}
\\ \\
Q \ar[rr]^{\txt{\small{$\varphi_{Q}$}}} & & Q }
\]

\begin{definition}
Let $\varphi:\rkq\longto\rkq$ be a $\pi_Q$-bundles
 morphism in the above conditions. The {\rm
canonical prolongation of the diffeomorphism $\varphi$} is the map
$j^{1*}\varphi:J^1\pi_Q\longto J^1\pi_Q$ given by
\[(j^{1*}\varphi)(j^1_q\sigma):=
j^1_{\varphi_Q(q)}(\varphi\circ \sigma\circ \varphi_Q^{-1})\,,\]
where $\sigma=(\sigma_{\rk},Id_Q)$ and
$\sigma_{\rk}:Q\stackrel{\sigma}{\longto}\rkq\stackrel{\pi_{\rk}}{\longto}\rk$.
\end{definition}

It is clear that this definition is valid, because choosing other
representative $\sigma '$ with the same $1$-jet at $q$ gives the
same result, that is, $j^{1*}\varphi(j^1_q\sigma)$ is well defined.

In local coordinates, if
$\varphi(t^B,q^j)=(\varphi^A(t^B,q^j),\varphi^i_Q(q^j))$ then
$$
j^{1*}\varphi(t^B, q^j,p^B_j)=\left(\varphi^A(t^B,
q^j),\varphi_Q^j(q),\left(\ds\frac{\partial \varphi^A}{\partial
q^k}+p^B_k\ds\frac{\partial \varphi^A}{\partial
t^B}\right)\ds\frac{\partial (\varphi_Q^{-1})^k}{\partial
q^i}\Big\vert_{\varphi_Q(q^j)}\right)\,.
$$

\begin{definition}
Let $Z\in\mathfrak{X}(\rkq)$ be a $\pi_{Q}$-projectable vector
field, with local $1$-parameter group of transformations
$\varphi_s:\rkq\longto \rkq$. Then the local $1$-parameter group
of transformations $j^{1*}\varphi_s:\rk\times
(T^1_k)^*Q\longto\rk\times (T^1_k)^*Q$ generates a vector field
$Z^{1*}\in\mathfrak{X}(\rk\times (T^1_k)^*Q)$, which is called the
{\rm complete lift of $Z$ to $\rk\times (T^1_k)^*Q$.}
\end{definition}

If the local expression of $Z\in\vf(\rkq)$ is
$Z=Z_A\ds\frac{\partial }{\partial
t^A}+Z^i\ds\frac{\partial}{\partial q^i}$, then
\[
Z^{1*}=Z^A\ds\frac{\partial }{\partial
t^A}+Z^i\ds\frac{\partial}{\partial q^i}+\left(\ds\frac{d
Z_A}{dq^i}-p^A_j\ds\frac{dZ^j}{dq^i}\right)\ds\frac{\partial
}{\partial p^A_i}\,,
\]
where $\ds\frac{d}{dq^i}$ denotes the total derivative, that is,
$\ds\frac{d}{dq^i}=\ds\frac{\partial }{\partial
q^i}+p^B_i\ds\frac{\partial }{\partial t^B}$.

\subsection{$k$-cosymplectic Hamiltonian systems}

Along this paper we are interested only in a kind of $k$-cosymplectic manifolds:
those which are of the form ${\cal M}=\rk\times M$, where
$(M,\omega_0^A,V)$ is a generic $k$-symplectic manifold.
Then,  denoting by
$$
\pi_{\rk}:\rk \times M \rightarrow \rk \quad , \quad
\pi_M:\rk \times M \rightarrow M
$$
 the canonical projections, we have the differential forms
$$
\eta^A=\pi_{\rk}^*\d t^A\quad , \quad \omega^A= \pi_M^*\omega^A_0\, ,
$$
and the distribution $V$ in $M$ defines a distribution ${\cal V}$ in
${\cal M}=\rk\times M$ in a natural way. All the conditions given in definition \ref{deest}
are verified, and hence $\rk\times M$ is a $k$-cosymplectic manifold.
From the Darboux Theorem \ref{dtk} we have local coordinates $(t^A,q^i,p^A_i)$ in $\rk\times M$.

Observe that the standard model is a particular class of these kinds of $k$-cosymplectic manifolds,
where $M=(T^1_k)^*Q$.

\begin{definition}
These kinds of $k$-cosymplectic manifolds
will be called {\rm almost-standard $k$-cosymplectic manifolds}.
\end{definition}

Consider an almost-standard $k$-cosymplectic manifold
$(\rk\times M,\eta^A,\omega^A,{\cal V})$, and let
$H\in\Cinfty( \rk\times M)$ be a Hamiltonian function.
The couple $(\rk\times M, H)$ is called a {\sl $k$-cosymplectic
Hamiltonian system}.

We denote by
$\vf^k_H(\rk\times M)$ the set of (local) $k$-vector fields
${\bf X}=(X_1,\dots,X_k)$ on $\rk\times M$
wich are solutions to the equations
\begin{equation}
\label{geonah}
 \eta^A(X_B)=\delta^A_B \quad , \quad \ds\sum_{A=1}^k \inn(X_A)\omega^A =
\d H-\ds\sum_{A=1}^k R_A( H) \eta^A \quad   .
\end{equation}
Since $R_A=\partial/\partial t^A$ and $\eta^A=\d t^A$, then
we can write locally  the above equations as follows
$$
\d t^A(X_B)=\delta^A_B \quad , \quad  \ds\sum_{A=1}^k\inn(X_A)\omega^A = \d  H-
\ds\sum_{A=1}^k\frac{\partial  H}{\partial t^A}\d t^A
\quad   .
$$

Furthermore, for a section $ \psi \colon  I \subset \rk\to \rk\times M$
of the projection $\pi_{\rk}$, the {\sl Hamilton-de Donder-Weyl equations}
 for this system are
\begin{equation}
 \ds\sum_{A=1}^k \inn\left(\psi_{*}(t)\left(\frac{\partial}{\partial t^A}\Big\vert_t\right)\right) (\omega^A\circ\psi) =
  \left[\d H-\ds\sum_{A=1}^k R_A( H) \eta^A\right]\circ\psi\  ,
\label{he0}
\end{equation}
 In Darboux coordinates, if
$ \psi(t)=( \psi^A(t), \psi^i(t), \psi^A_i(t))$,
as $\psi$ is a section of the projection $\pi_{\rk}$, it implies that
$\psi^A(t)= t^A$
the above equations leads to the equations
 \begin{equation}
 \label{he}
  \frac{\partial  H}{\partial q^i}=
-\ds\sum_{A=1}^k\frac{\partial \psi^A_i}{\partial t^A}
 \quad , \quad
 \frac{\partial  H}{\partial p^A_i}=
\frac{\partial \psi^i}{\partial t^A}\, .
\end{equation}

The relation between equations (\ref{geonah}) and (\ref{he0}) is given by the following:

\begin{theorem}
  Let ${\bf X}=(X_1,\dots,X_k)\in\vf^k_H(\rk\times M)$
 (i.e.; it is a $k$-vector field on
$\rk\times M$ which is a solution to the geometric Hamiltonian
equations (\ref{geonah})). If a section
  $\psi:\rk\to \rk\times M$  of $\pi_{\rk}$ is an integral section
of ${\bf X}$ then $\psi$ is a solution
to the Hamilton-de Donder-Weyl field equations (\ref{he0}).
\label{maint}
\end{theorem}
\proof \
 Let ${\bf X}=(X_1,\dots,X_k)\in\vf^k_H(\rk\times M)$ be locally given by
$$
X_A=(X_A)^B\frac{\partial}{\ds\partial t^B} +
(X_A)^i\frac{\partial}{\ds\partial q^i}
+(X_A)^B_i\frac{\partial}{\ds\partial p^B_i} \ ,
$$
then, from (\ref{geonah}) we obtain
\begin{equation}
(X_A)^B=\delta_A^B \quad , \quad
   \ds\frac{\partial H}{\partial p^A_i}=(X_A)^i \quad, \quad
    \ds\frac{\partial H}{\partial q^i}= -\ds\sum_{A=1}^k(X_A)^A_i \ ,
\label{111} 
\end{equation}
   and if $\psi:\rk\to\rk\times M$, locally given by
$\psi(t)=(t^A,\psi^i(t),\psi^A_i(t))$, is an integral section
of ${\bf X}$ , then
$$
 \ds\frac{\partial \psi^i}{\partial t^B}=(X_B)^i, \quad
\ds\frac{\partial \psi^A_i}{\partial t^B}=(X_B)^A_i \,  .
$$
Therefore, from (\ref{111}) we obtain that $\psi(t)$ is a solution
to the Hamiltonian field equations (\ref{he}).
\qed

And, conversely, we have:

\begin{lem}\label{int}
If  a section $\psi:\rk\to \rk\times M$ of $\pi_{\rk}$
is a solution to the Hamilton-de Donder-Weyl equation (\ref{he0})
 and $\psi$ is an integral section of ${\bf X}=(X_1,\dots,X_k)$, then
 ${\bf X}=(X_1,\dots,X_k)$ is solution to the equations
(\ref{geonah}) at the points of the image of $\psi$.
\end{lem}
\proof \
We must prove that
\begin{equation}\label{01}
    \ds\frac{\partial H}{\partial p^A_i}(\psi(t))=
  (X_A)^i(\psi(t)), \quad \ds\frac{\partial H}{\partial q^i}(\psi(t))=
   -\ds\sum_{A=1}^k(X_A)^A_i(\psi(t)) \ ,
\end{equation}
now as $\psi(t)=(t^A,\psi^i(t),\psi^A_i(t))$ is integral section of ${\bf X}$ we have that
\begin{equation}\label{02}
 \ds\frac{\partial \psi^i}{\partial t^B}(t)=(X_B)^i(\psi(t)), \quad
\ds\frac{\partial \psi^A_i}{\partial t^B}(t)=(X_B)^A_i(\psi(t)) \, .
\end{equation}

 As $\psi$ is a solution to the Hamilton-de Donder-Weyl equation (\ref{he})
then, from  (\ref{02}), we deduce (\ref{01}). \qed

We cannot claim that  ${\bf X}\in\vf^k_H(\rk\times M)$
because we cannot assure that ${\bf X}$ is a solution to the equations
 (\ref{geonah}) everywhere in $\rk\times M$.

\begin{prop}\label{he-geonah}
If  $\psi:U_0\subset\rk\longto \rk\times M$ is a  solution to the Hamilton-de
Donder-Weyl equation (\ref{he0}), then for each $t\in U_0$ there
exist a neighborhood $U_t$ of $t$ and a $k$-vector field
 ${\bf X}^t=(X^t_1,\dots,X^t_k)$ on $\psi(U_t)$ which is solution to the
equations (\ref{geonah}) in $\psi(U_t)$.
\end{prop}
\begin{proof}
 If $\psi:U_0\subset \r^k \to \rk\times M$ is a solution to the Hamilton-de
Donder-Weyl equation (\ref{he0}) then for every $t\in U_0$ there
exists a neigborhood $U_t\subset U_0$ of $t$, and a neigborhood
coordinate system $(W_t,s^A,q^i,p^A_i)$ of $\psi(t)$, such that
$\psi(U_t)=W_t\subset \psi(U_0)$, and
$\psi(s)=(s,\psi^i(s),\psi^A_i(s))$ for every $s\in U_t$.

As $\psi\vert_{U_t}:U_t\to W_t$ is an injective immersion
($\psi$ is a section and hence its image is an embedded submanifold),
we can define a $k$-vector field ${\bf X}^t=(X^t_1,\ldots,X^t_k)$
in $\psi(U_t)$  as follows
\[
X^t_A(\psi(s))=\psi_*(s)\Big(\frac{\partial}{\partial s^A}\Big\vert_{s}\Big)\quad , \quad  s\in U_t\, ,
\]
and so $\psi\vert_{U_t}$ is an integral section of ${\bf X}^t$. Then,
from the Lemma \ref{int} one obtains that ${\bf X}^t$
is solution to the equations (\ref{geonah}) in $\psi(U_t)$.
\end{proof}

{\bf Remark}:
It should be noticed that, in general, equations (\ref{geonah}) do
not have a single solution. In fact, if
$({\cal M},\eta^A,\omega^A,{\cal V})$ is a
k-cosymplectic manifold we can define the vector bundle morphism
$$
\begin{array}{rccl}
\omega^{\sharp}: & T^1_k{\cal M} & \longrightarrow & T^*{\cal M}  \\
& (X_1,\dots,X_k) & \mapsto &
\omega^{\sharp}(X_1,\dots,X_k) =  \displaystyle \sum_{A=1}^k \,
\inn(X_A)\omega^A
\end{array}
$$
and, denoting by $ {\mathbb M}_k(\mathbb{R})$ the space of
matrices of order $k$ whose entries are real numbers, the vector
bundle morphism
$$
\begin{array}{rccl}
\eta^{\sharp}: & T^1_k{\cal M} & \longrightarrow &
 {\cal M} \times {\mathbb M}_k(\mathbb{R})  \\
\noalign{\medskip} & (X_1,\dots,X_k) & \mapsto &
(\tau(X_1, \ldots , X_k),\eta^A(X_B))\, .
\end{array}
$$
We denote by the same symbols $\omega^{\sharp},\eta^{\sharp}$
 their natural extensions to vector fields and forms.

Now, let $H: {\cal M} \to \mathbb{R}$ be a real
$C^{\infty}$-function on $M$. Then, as in the case of an
almost-standard $k$-cosymplectic manifold, we can consider the set
${\frak X}^{k}_{H}{\cal M}$ of the (local) $k$-vector fields ${\bf
X} = (X_1, \ldots , X_k)$ on ${\cal M}$ which are solutions to the
equations
\begin{equation}\label{HamEqs}
\eta^{A}(X_{B}) = \delta_{B}^{A}, \; \; \displaystyle
\sum_{A=1}^{k} i(X_{A})\omega^A = dH - \sum_{A=1}^{k}
R_{A}(H)\eta^A.
\end{equation}
Moreover, we may prove the following result
\begin{prop}
The solutions of Eqs. (\ref{HamEqs}) are the sections of an affine
bundle of rank $(k-1)(kn+n)$ which is modelled on the vector
subbundle $\ker\omega^{\sharp} \cap \ker \eta^{\sharp}$ of
$T_{k}^{1}{\cal M}$.
\end{prop}
\begin{proof}
We consider the vector subbundle $\ker\eta^{\sharp}$ of
$T_k^1{\cal M}$ and the vector bundle morphism
\[
\omega^{\sharp}\vert_{\ker \eta^{\sharp}}: \ker \eta^{\sharp} \to
T^*{\cal M}.
\]
It is clear that this morphism takes values in the vector
subbundle $\displaystyle \cap_{A=1}^k \langle R_{A}\rangle^{0}$ of $T^*{\cal M}$,
where $\langle R_{A}\rangle^{0}$ is the vector subbundle of $T^*{\cal M}$ whose
fiber at the point $x\in {\cal M}$ is
$\{ \alpha \in T^*_x{\cal M} / \alpha(R_A(x)) = 0\}$.
Furthermore, we have that
\[
\ker (\omega^{\sharp}\vert_{\ker \eta^{\sharp}}) = \ker\omega^{\sharp}
\cap \ker\eta^{\sharp}.
\]
We will prove that
\[
\omega^{\sharp}\vert_{\ker \eta^{\sharp}}: \ker \eta^{\sharp} \to
\displaystyle \cap_{A=1}^{k} \langle R_{A}\rangle^0
\]
is an epimorphism of vector bundles.
For this purpose, we will see that the dual morphism
\[
(\omega^{\sharp}\vert_{\ker \eta^{\sharp}})^*: \displaystyle
(\cap_{A=1}^{k} \langle R_{A}\rangle^0)^* \to (\ker \eta^{\sharp})^*
\]
is a monomorphism of vector bundles.

First of all, it is clear that the dual bundle to $\cap_{A=1}^{k}
\langle R_{A}\rangle^0$ (respectively, $\ker\eta^{\sharp}$) may be identified
with the vector bundle whose fiber at the point $x\in {\cal M}$ is
$\displaystyle \cap_{A=1}^k\langle \eta^A(x)\rangle^0$ (respectively,
$\{(\alpha_{1}, \ldots , \alpha_{k}) \in ((T^1_k)^*{\cal M})_x /
\alpha_{A}(R_{B}(x)) = 0, \mbox{ for all } A, B \}$). Under these
identifications, the morphism $\omega^{\sharp}_{|\ker
\eta^{\sharp}}$ is given by
\[
(\omega^{\sharp}\vert_{\ker \eta^{\sharp}})^*(v) = (i(v)\omega^1(x),
\ldots , i(v)\omega^k(x)), \; \; \mbox{ for } v\in \displaystyle
\cap_{A=1}^k \langle \eta^A(x)\rangle^0.
\]
Thus, $(\omega^{\sharp}\vert_{\ker \eta^{\sharp}})^*$ is a monomorphism
of vector bundles.
Then $\omega^{\sharp}_{\ker\eta^{\sharp}}: \ker
\eta^{\sharp} \to \displaystyle \cap_{A=1}^k \langle R_{A}\rangle^0$ is an
epimorphism of vector bundles.

So, as the rank of the vector bundle $\ker \eta^{\sharp}$
(respectively, $\displaystyle \cap_{A=1}^k \langle R_{A}\rangle^0$) is $k(kn +
n)$ (respectively, $kn + n$), we deduce that the rank of the
vector bundle $\ker \omega^{\sharp} \cap \ker \eta^{\sharp}$ is
$(k-1)(kn+n)$.

Furthermore, if $(X_1, \ldots , X_k)$ is a particular
solution of the Eqs. (\ref{geonah}) and ${\bf Z}$ is a section of
the vector bundle $\ker\omega^{\sharp} \cap \ker\eta^{\sharp} \to
{\cal M}$ then $(X_1, \ldots , X_k) + {\bf Z}$ also is a solution
of these equations. In addition, if ${\bf X'}$ and ${\bf X}$ are
solutions of Eqs. (\ref{geonah}) then ${\bf Z} = {\bf X'} - {\bf
X}$ is a section of the vector bundle $\ker\omega^{\sharp} \cap
\ker\eta^{\sharp} \to {\cal M}$.

Finally, if $(t^A, q^i, p^A_i)$ are Darboux coordinates in a
neighborhood $U_x$ of each point $x \in {\cal M}$ then we may
define a local $k$-vector field on $U_x$ that satisfies
(\ref{HamEqs}). For instance, we can put
$$
 (X_1)^1_i=
\ds\frac{\partial H}{\partial q^i}\quad , \quad (X_A)^B_i=0
\,\,\,\mbox{for $A\neq 1\neq B$} \quad , \quad (X_A)^i=
\ds\frac{\partial H}{\partial p^A_i}\, .
$$
Now one can construct a global $k$-vector field, which is a
solution of (\ref{geonah}), by using a partition of unity in the
manifold ${\cal M}$ (see \cite{mod1}).
\end{proof}

\section{Symmetries for $k$-cosymplectic Hamiltonian systems}
\protect\label{ntmvf}

\subsection{Symmetries and conservation laws}

Let $(\rk\times M, H)$ be a $k$-cosymplectic Hamiltonian system.
First, following \cite{Olver}, we introduce the next definition :

\begin{definition}
\label{Olver} A {\rm conservation law} (or a {\rm conserved quantity})
for the Hamilton-de Donder-Weyl equations (\ref{he0})
is a map  ${\cal F}=({\cal F}^1 , \ldots , {\cal F}^k)\colon
\rk\times M \longto \r^k$ such that the divergence of
$$
{\cal F}\circ \psi=({\cal F}^1 \circ \psi, \ldots , {\cal F}^k \circ
\psi)\colon U_0\subset\r^k \longto \r^k
$$
is zero for every solution  $\psi$
 to the Hamilton-de Donder-Weyl equations (\ref{he0}); that
is for all $t\in U_0\subset \rk$,
\begin{equation}\label{cons-law}
0 = [Div({\cal F}\circ\psi)](t)=
\sum_{A=1}^k \frac{\partial ({\cal F}^A \circ \psi)}{\partial t^A}\Big\vert_{t}=
\sum_{A=1}^k\psi_*(t)\Big(\frac{\partial}{\partial t^A}\Big\vert_{t}\Big)({\cal F}^A)
\;.
\end{equation}
\end{definition}

\begin{prop}
The map ${\cal F}=({\cal F}^1,\ldots,{\cal F}^k)\colon \rk\times M\longto \r^k$
 defines a conservation law if, and only if,  for
  every integrable $k$-vector field ${\bf X}=(X_1,\dots,X_k)$
which is a solution to the equations (\ref{geonah}), we have that
\begin{equation}\label{cons-law field}
\sum_{A=1}^k\Lie({X_A}){\cal F}^A=0\;.
\end{equation}
\end{prop}
\begin{proof}
 (\ref{cons-law})$\Rightarrow$ (\ref{cons-law field})\
Let ${\cal F}=({\cal F}^1,\ldots,{\cal F}^k)$ be a conservation
law and ${\bf X}=(X_1,\dots,X_k)\in\vf^k_H(\rk\times M)$ an
integrable $k$-vector field. If  $\psi\colon U_0\subset
\r^k\longto\rk\times M$ is  an integral section of ${\bf X}$, by
the Lemma \ref{int}, we have that $\psi$ is a solution to the
Hamilton-de Donder-Weyl equation (\ref{he0}), and by definition of
integral section we have that $
X_A(\psi(t))=\psi_*(t)\left(\derpar{}{t^A}\Big\vert_t\right) $.
Therefore from (\ref{cons-law}) we obtain (\ref{cons-law field}).

Conversely, (\ref{cons-law field})$\Rightarrow$ (\ref{cons-law}).\
In fact, we must prove that
for every solution  $\psi\colon U_0 \to \rk\times M$
 to the Hamilton-de Donder-Weyl equations (\ref{he0})
 the identity  (\ref{cons-law}) holds.
From Proposition \ref{he-geonah} there exist a $k$-vector field
 ${\bf X}=(X_1,\dots,X_k)$ on $\psi(U_0)$
 which is solution to the equations (\ref{geonah}) and $\psi$
 is an integral section of ${\bf X} $. We know that
\[
X_A(\psi(t))=\psi_*(t)\Big(\frac{\partial}{\partial t^A}\Big\vert_{t}\Big)\quad , \quad t\in U_0\, .
\]
Then  for all $\psi(t)\in \psi(U_0)$
$$
0=\sum_{A=1}^k\Lie(X_A){\cal F}^A (\psi(t) )
 =\sum_{A=1}^k\psi_*(t)\left(\derpar{}{t^A}\Big\vert_t\right)({\cal
F}^A)
$$
\end{proof}

\begin{definition}
\label{symH}
\begin{enumerate}
\item
A {\rm symmetry} of the $k$-cosymplectic Hamiltonian system
$(\rk\times M,H)$ is a diffeomorphism $\Phi\colon \rk\times M\longto \rk\times M $
verifying the following conditions:
\begin{enumerate}
\item
It is a fiber preserving map for the trivial bundle
$\pi_{\rk}\colon\rk\times M\to\rk$; that is, $\Phi$ induces a diffeomorphism
$\phi\colon\rk\to\rk$ such that
$\pi_{\rk}\circ\Phi=\phi\circ\pi_{\rk}$.
\item
For every section
$\psi$ solution to the Hamilton-de Donder-Weyl equations (\ref{he0}),
we have that the section $\Phi\circ\psi\circ\phi^{-1}$ is also a solution to these equations.
\end{enumerate}
\item
An {\rm infinitesimal symmetry}
of the $k$-cosymplectic Hamiltonian system $(\rk\times M,H)$
is a vector field $Y\in\vf(\rk\times M)$
whose local flows are local symmetries.
\end{enumerate}
\end{definition}

As a consequence of the definition, all the results that we state
for symmetries also hold for infinitesimal symmetries.

Symmetries  can be used to generate new conservation laws from
a given conservation law, In fact, a first straightforward consequence of definitions \ref{Olver} and
\ref{symH} is:

\begin{prop}
If $\Phi\colon \rk\times M\longto \rk\times  M$ is
a symmetry of a $k$-cosymplectic
 Hamiltonian system and
 ${\cal F}=({\cal F}^1, \ldots , {\cal F}^k)\colon \rk\times M\longto \r^k$
is a conservation law, then so is
 $\Phi^*{\cal F}=(\Phi^*{\cal F}^1,\ldots,\Phi^*{\cal F}^k)$.
\label{generador}
 \end{prop}
 \proof\
 For every section
$\psi$ solution to the Hamilton-de Donder-Weyl equations and for every $t\in\rk$,
we have that
$$
(\Phi^*{\cal F}\circ\psi)(t)=({\cal F}\circ\Phi\circ\psi)(t)=
({\cal F}\circ\Phi\circ\psi\circ\phi^{-1}\circ\phi)(t)=
({\cal F}\circ\Phi\circ\psi\circ\phi^{-1})(\phi(t))\ ,
$$
and therefore
$$
Div(\Phi^*{\cal F}\circ\psi)=0\ \Leftrightarrow\ Div({\cal F}\circ\Phi\circ\psi\circ\phi^{-1})=0
$$
on the corresponding domains. But
the last equality holds since ${\cal F}$ is a conservation law and
$\Phi\circ\psi\circ\phi^{-1}$ is also a solution to the Hamilton-de Donder-Weyl equations.
 \qed

The following proposition gives a characterization of symmetries
in terms of $k$-vector fields.

\begin{prop}
Let $(\rk\times M,H)$ be a $k$-cosymplectic Hamiltonian system and
$\Phi\colon\rk\times M\to\rk\times M$ a fiber preserving diffeomorphism
for the trivial bundle $\pi_{\rk}\colon\rk\times M\to\rk$.
 \ben
\item
 For every integrable $k$-vector field
${\bf X}=(X_1,\dots,X_k)$ and  for every integral section  $\psi$ of ${\bf X}$,
the section $\Phi\circ\psi\circ\phi^{-1}$ is an integral section of the $k$-vector field
 $\Phi_*{\bf X}=(\Phi_*X_1,\dots,\Phi_*X_k)$,
 and hence $\Phi_*{\bf X}$ is integrable.
 \item
 $\Phi$ is a symmetry if, and only if, for every integrable $k$-vector field
 ${\bf X}=(X_1,\dots,X_k)\in\vf^k_H(\rk\times M)$, then
$\Phi_*{\bf X}=(\Phi_*X_1,\dots,\Phi_*X_k)\in\vf^k_H(\rk\times M)$.
\een
\label{pro4}
\end{prop}
\proof\
\ben
\item
Given $x\in \rk\times M$, let $\psi\colon U_0\subset \r^k \to \rk\times M$
 an integral section of ${\bf X}$ passing through $x$; that is $\psi(0)=x$, then
 $\Phi\circ\psi\circ\phi^{-1}\colon\phi(U_0)\subset \r^k \to \rk\times M$
  is a section passing through $\Phi(x)$; that is,
 if $t_0=\phi(0)$, then $(\Phi\circ\psi\circ\phi^{-1})(t_0)=\Phi(x)$.
 
  Next we have to prove that $\Phi\circ\psi\circ\phi^{-1}$
 is an integral section of $\Phi_*{\bf X}$; that is,
 for every $t\in\phi(U_0)$, and for every $A=1,\dots,k$,
 $$
 (\Phi\circ\psi\circ\phi^{-1})_*(t)\left(\derpar{}{t^A}\Big\vert_t\right)=
 (\Phi_*X_A)((\Phi\circ\psi\circ\phi^{-1})(t)) \ ,
 $$
 or, what is equivalent, that the following diagram is commutative
$$
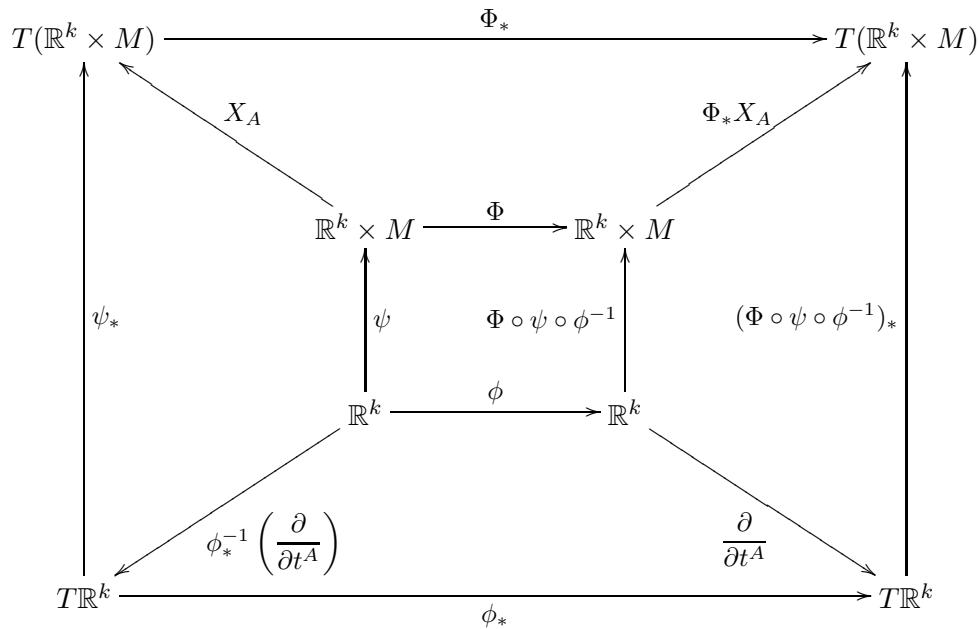
\xymatrix{
T(\rk\times M)
\ar[rrrrrr]^{\txt{\small{$\Phi_*$}}} & & & & & & T(\rk\times M)
\\ \\
& &
\ar[uull]_{\txt{\small{$X_A$}}}
 \rk\times M
\ar[rr]^(.5){\txt{\small{$\Phi$}}}
& &\rk\times M
\ar[uurr]^{\txt{\small{$ \Phi_*X_A$}}} & &
\\ \\
 & & \rk 
 \ar[uu]_{\txt{\small{$\psi$}}}
  \ar[ddll]^{\txt{\small{$\phi^{-1}_*\left(\derpar{}{t^A}\right)$}}}
 \ar[rr]^{\txt{\small{$\phi$}}} & & \rk 
 \ar[uu]^{\txt{\small{$\Phi\circ\psi\circ\phi^{-1}$}}} 
\ar[ddrr]_ {\txt{\small{$\derpar{}{t^A}$}}}
\\  \\
T\rk
\ar[rrrrrr]_ {\txt{\small{$\phi_*$}}}
\ar[uuuuuu]_{\txt{\small{$\psi_*$}}}  & & & & & &
\ar[uuuuuu]^{\txt{\small{$\,(\Phi\circ\psi\circ\phi^{-1})_*$}}} 
 T\rk
}
$$
 First, we must take into account that the diffeomorphism $\phi\colon\rk\to\rk$
 makes a change of global coordinates in $\rk$; that is,
 $\phi(\tilde t^A)=(t^A)$, and then, if $\psi$
 is an integral section of ${\bf X}$, we have that
 \beann
 \psi_*(\phi^{-1}(t))\left(\derpar{}{\tilde t^A}\Big\vert_{\phi^{-1}(t)}\right)&=&
  \psi_*(\phi^{-1}(t))\left(\phi^{-1}_*(t)\left(\derpar{}{t^A}\Big\vert_t\right)\right)\\ &=&
X_A(\psi\circ\phi^{-1})(t)) \ .
\eeann
  Then, we obtain
\beann
(\Phi\circ \psi\circ\phi^{-1})_*(t)\left(\derpar{}{t^A}\Big\vert_t\right) &=&
\Phi_*(\psi(\phi^{-1}(t))\left(\psi_*(\phi^{-1}(t))\left(\phi^{-1}_*(t)
\left(\derpar{}{t^A}\Big\vert_t\right)\right)\right)
  \\ &=&
  \Phi_*(\psi(\phi^{-1}(t))\left(X_A(\psi(\phi^{-1}(t))\right) \\ &=& 
  (\Phi_*X_A)((\Phi\circ\psi\circ\phi^{-1})(t))  \ . 
\eeann
\item
   $( \Rightarrow )$ Now,
let $x$ be an arbitrary point of $\rk \times M$ and $\psi$ be an
integral section of ${\bf X}$
 passing trough the point $\Phi^{-1}(x)$, that is $\psi(0)=\Phi^{-1}(x)$.
 We know that
   $\psi$ is a solution to the Hamilton-de Donder-Weyl equations (\ref{he0}).
 Since $\Phi$ is a symmetry, $\Phi\circ \psi\circ\phi^{-1}$ is a solution
 to the Hamilton-de Donder-Weyl equations (\ref{he0}) and,
 by the item 1,
it is an integral section of $\Phi_* {\bf X}$ passing trough
the point $ \Phi(\psi(0))=\Phi(\Phi^{-1}(x))=x$
(this means that $(\Phi\circ \psi\circ\phi^{-1})(\phi(0)= (\Phi\circ\psi)(0))=x$).
 Hence, from Lemma
\ref{int}, we deduce that $\Phi_* {\bf X}\in\vf^k_H(\rk\times M)$ at the
points $(\Phi\circ \psi)(t)$,
 in particular at the arbitrary point $(\Phi\circ\psi)(0)=x $.

$( \Leftarrow )$ Conversely, let $\psi\colon U_o\subset\rk\to\rk\times M$ be a solution
 to the Hamilton-de Donder-Weyl equations (\ref{he0},
 then (see Proposition \ref{he-geonah})
   there exists a $k$-vector field
${\bf X}=(X_1,\dots,X_k)$ on $\psi(U_0)$
 which is solution to the equations (\ref{geonah}) and $\psi$
 is an integral section of ${\bf X} $ in $\psi(U_0)$.

Since ${\bf X}$ is solution to (\ref{geonah}), then
 $\Phi_*{\bf X}=(\Phi_*X_1,\dots,\Phi_*X_k)
\in\vf^k_H(\rk\times M)$ by hypothesis, and then,
as a consequence of the item 1 and theorem \ref{maint}, $\Phi \circ \psi\circ\phi^{-1}$
 is a solution to the Hamilton-de Donder-Weyl equations (\ref{he0}).
\qed
\een

As a consequence of this,  if $\Phi$ is a symmetry and
${\bf X}$ is an integrable $k$-vector field in $\vf^k_H(\rk\times M)$,
we have that
 $\Phi_*{\bf X}-{\bf X}\in\ker\,\omega^{\sharp}\cap \ker\,\eta^{ \sharp}$.

\begin{prop}
Let $(\rk\times M,H)$ be a
$k$-cosymplectic Hamiltonian system.
If \,$Y\in{\mathfrak{X}}(\rk\times M)$ is an
infinitesimal symmetry, then for every integrable $k$-vector field
${\bf X}=(X_1,\dots,X_k)\in\vf^k_H(\rk\times M)$
we have that $[Y,{\bf X}]=([Y,X_1],\dots,[Y,X_k])\in
\ker\,\omega^{ \sharp}\cap \ker\,\eta^{ \sharp}  $.
\label{pro5}
\end{prop}
\proof\
Denote by $F_t$ the local $1$-parameter groups of diffeomorphisms
generated by $Y$. As $Y$ is an infinitesimal symmetry,
as a consequence of Proposition \ref{pro4} we have
$F_{t*}{\bf X}-{\bf X}={\bf Z}\in
\ker\,\omega^{ \sharp}\cap \ker\,\eta^{ \sharp}  $.
 Then, taking a local basis of sections
$\{{\bf Z}^1,\ldots,{\bf Z}^r\}=
\{(Z^1_1,\ldots,Z^1_k),\ldots,(Z^r_1,\ldots,Z^r_k)\}$ of the
vector bundle $\ker\,\omega^{ \sharp}\cap \ker\,\eta^{ \sharp} \to
\rk \times M$, we have that $F_{t*}{\bf X}-{\bf X}=g_\alpha {\bf
Z}^\alpha$, $\alpha=1,\ldots,r$, with
$g_\alpha\colon\Real\times(\rk\times M)\to\Real$ (they are
functions that depend on $t$, some of them different from $0$);
that is
$$
F_{t*}{\bf X}-{\bf X}=(F_{t*}X_1-X_1,\ldots,F_{t*}X_k-X_k)=
(g_\alpha Z_1^\alpha,\ldots,g_\alpha Z_k^\alpha)=
g_\alpha {\bf Z}^\alpha \ .
$$
 Therefore
\beann
[Y,{\bf X}]&=&\Lie(Y){\bf X}=(\Lie(Y)X_1,\ldots\Lie(Y)X_k)=
\left( \lim_{t\to 0}\frac{F_{t*}X_1-X_1}{t},\ldots,
 \lim_{t\to 0}\frac{F_{t*}X_k-X_k}{t} \right) \\
&=&
\left( \lim_{t\to 0}\frac{g_\alpha}{t}Z_1^\alpha,\ldots,
 \lim_{t\to 0}\frac{g_\alpha}{t}Z_k^\alpha \right)=
(f_\alpha Z_1^\alpha,\ldots,f_\alpha Z_k^\alpha)=
f_\alpha {\bf Z}^\alpha\in\ker\,\omega^{ \sharp}\cap \ker\,\eta^{ \sharp}   \ ,
\eeann
where $f_\alpha\colon \rk\times M\to\Real$.
\qed

\subsection{$k$-cosymplectic Noether symmetries. Noether's theorem}
\label{kcnsNt}

As it is well known, the existence of symmetries
is associated with the existence of conservation laws.
How to obtain these conservation laws depends
on the symmetries that we are considering.
In particular, for Hamiltonian and Lagrangian systems,
Noether`s theorem gives a rule for doing it,
for certain kinds of symmetries:
those that preserve both the physical information
(given by the Hamiltonian or the Lagrangian function),
and some geometric structures of the system.
For $k$-cosymplectic Hamiltonian field theories
a reasonable choice consists in taking
those symmetries preserving the $k$-cosymplectic
structure as well as the Hamiltonian function.
Bearing this in mind, first we prove the following:

\begin{prop}
Let $(\rk\times M,H)$ be a $k$-cosymplectic Hamiltonian system.
\begin{enumerate}
 \item
If   $\Phi\colon \rk\times M\longto \rk\times M$ is a diffeomorphism satisfying that

(a)\quad $\Phi^*\omega^A=\omega^A$, \quad
(b)\quad $\Phi^*\eta^A=\eta^A$, \quad
(c)\quad $\Phi^*H=H$,

 then $\Phi$ is a symmetry of the $k$-cosymplectic
 Hamiltonian system $(\rk\times M,H)$.
\item
If $Y\in\vf(\rk\times M)$ a vector field satisfying that

(a)\quad $\Lie(Y)\omega^A=0$, \quad
(b)\quad $\Lie(Y)\eta^A=0$, \quad
(c)\quad $\Lie(Y)H=0$,

then $Y$ is an infinitesimal symmetry of the $k$-cosymplectic Hamiltonian system  $(\rk\times M,H)$.
\end{enumerate}
\label{previa}
\end{prop}
\begin{proof}
\begin{enumerate}
 \item
First, from
$$
\d t^A=\eta^A=\Phi^*\eta^A=\Phi^*\d t^A=\d\Phi^*t^A
$$
we conclude that $\Phi^*t^A=t^A+k^A$ ($k^A\in\Real$). 
This result (together with the condition $\Phi^*\omega^A=\omega^A$)
means that the local expression of $\Phi$ is
$\Phi(t^A,q^i,p_i)=(t^A+k^A,\Phi^i(q,p),\Phi^A_i(q,p))$.
Therefore it induces a diffeomorphism $\phi\colon\rk\to\rk$
given by $\phi(t^A)=t^A+k^A$; hence $\pi_{\rk}\circ\Phi=\phi\circ\pi_{\rk}$
and $\Phi$ is a fiber preserving map for the trivial bundle
$\pi_{\rk}\colon\rk\times M\to\rk$.

Now, as
 $\ds \eta^A(R_B)=\d t^A\left(\derpar{}{t^B}\right)=\delta^A_B$,
and $\Phi^*\eta^A=\Phi^*\d t^A=\d t^A=\eta^A$, we have
\[
\delta^A_B=\d t^A\left(\derpar{}{t^B}\right)=
(\Phi^*\d t^A)\left(\derpar{}{t^B}\right)=
\Phi^*\left\{\d t^A\left(\Phi_*\left(\derpar{}{t^B}\right)\right)\right\} \ ,
\]
thus
\[
 \Phi_*\left(\derpar{}{t^B}\right)=
\derpar{}{t^B} +\alpha^i \derpar{}{q^i}+\beta^A_i \derpar{}{p^A_i}\  ,
\]
but, since $\Phi^*\omega^A = \omega^A$, for all $A$,
\[
0 = i\left(\displaystyle \frac{\partial}{\partial t^B}\right)\omega^A =
i\left(\Phi_*\left(\frac{\partial}{\partial t^B}\right)\right)(\omega^A \circ \Phi)
\]
and then
\[
\displaystyle \frac{\partial}{\partial t^B} + \alpha^i
\frac{\partial}{\partial q^i} + \beta_{i}^A
\frac{\partial}{\partial p^A_i} = \Phi_*\left(\frac{\partial}{\partial
t^B}\right) \in \cap_{A=1}^k \ker(\omega^A \circ \Phi) =
\left\langle \frac{\partial}{\partial t^A} \circ \Phi\right\rangle_{A=1,\ldots , k}
\]
which implies that $\displaystyle \Phi_*\left(\frac{\partial}{\partial
t^B}\right) = \frac{\partial}{\partial t^B}$ that is, $\Phi_*(R_B) =
R_B$.

Furthermore, for every $k$-vector field
${\bf X}=(X_1,\ldots, X_k)\in \vf^k_H(\rk\times M)$,
we obtain that
$$
\Phi^*(\eta^A (\Phi_*X_B))= (\Phi^*\eta^A)(X_B)=\eta^A(X_B)=\delta^A_B  \ ,
$$
$$
\begin{array}{l}
\Phi^*\left[\ds\sum_{A=1}^k\inn(\Phi_*X_A)\omega^A-\d
H+(\Lie(R_A)H)\eta^A\right]=
 \\ \noalign{\medskip}
\sum_{A=1}^k\left[\inn(X_A)(\Phi^*\omega^A)-\Phi^*\d H+(\Phi^*\Lie(R_A) H) (\Phi^*\eta^A) \right]=
\\ \noalign{\medskip}
\sum_{A=1}^k[\inn(X_A)\omega^A-\d H+(\Lie(R_A)H) \eta^A]=0\ .
 \end{array}
 $$
Hence, as $\Phi$ is a diffeomorphism, these results are equivalent
to demanding that
$$
\eta^A(\Phi_*X_B)=\delta^A_B \ ; \
\sum_{A=1}^k[\inn(\Phi_*X_A)\omega^A-\d H+\Lie(R_A) H\eta^A]=0 \ .
$$
 Thus
$\Phi_*{\bf X}=(\Phi_*X_1,\dots,\Phi_*X_k)\in\vf^k_H(\rk\times M)$.
Finally, if ${\bf X}$ is integrable, then
$\Phi_*{\bf X}$ is integrable too (as Proposition \ref{pro4} claims),
and thus $\Phi$ is a symmetry.
\item
It is a consequence of the above item, taking the local flows of $Y$.
\end{enumerate}
\end{proof}

Although the condition $2(b)$ of the hypothesis is sufficient to prove that
 these kinds of vector fields are infinitesimal symmetries,
 in order to achieve a good generalization of Noether's theorem,
this condition must be hardered by demanding that
$\inn(Y)\eta^A=0$ (observe that $\inn(Y)\d t^A=0$
$\Longrightarrow$ $\Lie(Y)\d t^A=0$).
This is equivalent to write $\Lie(Y)t^A=0$ and hence,
the equivalent global condition $1(b)$ for this case
is $\Phi^*t^A=t^A$.This means that the induced diffeormorphism
 $\phi\colon\rk\to \rk$ is the identity on $\rk$.

 Taking into account all of this, we introduce the following definitions:

\begin{definition}
\label{ksns}
 Let $(\rk\times M,H)$ be a $k$-cosymplectic  Hamiltonian system.
 \begin{enumerate}
 \item
 A {\rm $k$-cosymplectic Noether symmetry} is a diffeomorphism
$\Phi\colon \rk\times M\longto \rk\times M$
satisfying the following conditions:

(a)\quad $\Phi^*\omega^A=\omega^A$, \quad
(b)\quad $\Phi^*t^A=t^A$, \quad
(c)\quad $\Phi^*H=H$.

If the $k$-symplectic structure is exact,
a $k$-cosymplectic Noether symmetry
is said to be {\rm exact} if $\Phi^*\theta^A=\theta^A$.

In the particular case that $M=(T^1_k)^*Q$ (the standard model),
if $\Phi=j^{1*}\varphi$ for some
diffeomorphism $\varphi\colon \rkq\longto \rkq$,
then the $k$-cosymplectic Noether symmetry
 $\Phi$ is said to be {\rm natural}.
\item
Let $(\rk\times M,H)$ be a $k$-cosymplectic  Hamiltonian system.
An {\rm infinitesimal $k$-cosymplectic Noether symmetry}
is a vector field $Y\in\vf(\rk\times M)$ whose local flows are local $k$-cosymplectic Noether symmetries;
that is, it satisfies that:

(a)\quad $\Lie(Y)\omega^A=0$, \quad
(b)\quad $\inn(Y)\eta^A=0$, \quad
(c)\quad $\Lie(Y)H=0$.

If the $k$-symplectic structure is exact,
an infinitesimal $k$-cosymplectic Noether symmetry
is said to be {\rm exact} if $\Lie(Y)\theta^A=0$.

In the particular case that $M=(T^1_k)^*Q$,
if $Y=Z^{1*}$ for some $Z\in\vf (\rkq)$, then the
infinitesimal $k$-cosymplectic Noether symmetry $Y$
is said to be {\rm natural}.
\end{enumerate}
\label{CNsym}

(Obviously natural (infinitesimal) $k$-cosymplectic Noether symmetries are exact).
\end{definition}

\begin{lem}
\label{lematec}
If $Y\in\vf(\rk\times M)$
is an infinitesimal $k$-cosymplectic Noether symmetry,
then $[Y,R_A]=~0$.
\end{lem}
\begin{proof}
In fact, for all $A,B$, we have that
\beann
\inn([Y,R_A])\omega^B =
\Lie(Y)\inn(R_A)\omega^B-
\inn(R_A)\Lie(Y)\omega^B=0 & \Longrightarrow &
[Y,R_A]\in\ker\,\omega^B\ ,
\\
\inn([Y,R_A])\eta^B =
\Lie(Y)\inn(R_A)\eta^B-\inn(R_A)\Lie(Y)\eta^B= \Lie(Y)\delta_A^B=0
& \Longrightarrow & [Y,R_A]\in\ker\,\eta^B\ , \eeann and then
$[Y,R_A]\in(\cap_B\ker\,\omega^B)\cap(\cap_B\ker\,\eta^B)=\{ 0\}$.
\end{proof}

{\bf Remarks}:
\begin{itemize}
 \item
The condition $\Phi^*t^A=t^A$ means that
$k$-cosymplectic Noether symmetries generate transformations along
the fibres of the projection
 $\pi_{\rk}\colon \rk\times M\longrightarrow\rk$; that
is, they leave the fibres of the projection
 $\pi_{\rk}\colon \rk\times M\longrightarrow\rk$ invariant or,
 what means the same thing,
$\pi_{\rk}\circ \Phi=\pi_{\rk}$.

As a consequence,  in the particular case that $M=(T^1_k)^*Q$,
if $\Phi=j^{1*}\varphi$ (for some diffeomorphism $\varphi\colon
\rkq\longto \rkq$) is a natural $k$-cosymplectic Noether
symmetry,
 then the diffeomorphism $\varphi\colon \rkq\longto \rkq$
 must leave the fibres of the projection
$p_{\rk}\colon\rkq\longrightarrow\rk$ invariant necessarily;
that is, $p_{\rk}\circ\varphi=p_{\rk}$.
\item
In the case of infinitesimal $k$-cosymplectic Noether symmetries the
analogous condition is $\inn(Y)\d t^A=0$, which means that $Y$ has
the local expression
 $\ds Y=Y_i\ds\frac{\partial}{\partial q^i}+
Y^A_i\frac{\partial}{\partial p^A_i}$.
This means that $Y$ is tangent to the fibres of the projection
 $\pi_{\rk}\colon \rk\times M\longrightarrow\rk$. Thus
these infinitesimal symmetries only generate transformations along
these fibres, or, what means the same thing, the local flows of the
generators $Y$ leave the fibres of the projection
 $\pi_{\rk}\colon\rk\times M\longrightarrow\rk$ invariant.
 Furthermore, as a consequence of the above Lemma,
 and taking into account that $\ds R_A=\derpar{}{t^A}$,
in this local expression for $Y$
the component functions $Y_i,Y^A_i$ do not depend on the coordinates $(t^A)$.

Observe also that,  in the particular case that $M=(T^1_k)^*Q$,
if $Y=Z^{1*}$ (for some $Z\in\vf (\rkq)$) is a
natural infinitesimal $k$-cosymplectic Noether symmetry, then
$\inn(Y)\d t^A=0$, necessarily.
\end{itemize}

In addition, it is immediate to prove that, if $Y_1,Y_2\in\vf(\rk\times M)$
 are infinitesimal Noether symmetries, then so is $[Y_1,Y_2]$.

It is interesting to comment that, for infinitesimal
$k$-cosymplectic Noether symmetries, the results in the item 2 of
Proposition \ref{pro4} and in Proposition  \ref{pro5} hold, not
only for integrable $k$-vector fields in $\vf^k_H(\rk\times M)$,
but for every $k$-vector field ${\bf X}\in\vf^k_H(\rk\times M)$.
In fact, for the first one we have \beann
\sum_{A=1}^k\inn([Y,X_A])\omega^A &=&
\sum_{A=1}^k\{\Lie(Y)\inn(X_A)\omega^A-\inn(X_A)\Lie(Y)\omega^A\}=
\sum_{A=1}^k\Lie(Y)(\d H-(\Lie(R_A)H)\eta^A)
\\ &=&
\sum_{A=1}^k\{\d(\Lie(Y)(H))-(\Lie(Y)\Lie(R_A)H)\eta^A-(\Lie(R_A)H)\Lie(Y)\eta^A\}
\\ &=&
 -\sum_{A=1}^k(\Lie(R_A)\Lie(Y)H)\eta^A=0 \ .
\eeann
Furthermore
$$\inn([Y,X_A])\eta^B=\Lie(Y)\inn(X_A)\eta^B-\inn(X_A)\Lie(Y)\eta^B=0\, ,$$
and the proof for the second one is straighforward.

As infinitesimal $k$-cosymplectic Noether symmetries
are vector fields in $\rk\times M$
whose local flows are local $k$-cosymplectic Noether symmetries,
all the results that we state
for $k$-cosymplectic Noether symmetries also hold for
infinitesimal $k$-cosymplectic Noether symmetries.
Hence, from now on we consider only the infinitesimal case.

A first relevant result is the following:

\begin{prop}
Let $Y\in\vf(\rk\times M)$ be an
infinitesimal $k$-cosymplectic Noether symmetry.
Then, for every $p\in\rk\times M$, there is an open
neighbourhood $U_p\ni p$, such that:
\begin{enumerate}
\item
There exist ${\cal F}^A\in\Cinfty(U_p)$,
which are unique up to constant functions, such that
\begin{equation} \inn(Y)\omega^A=\d {\cal
F}^A, \qquad \mbox{\rm (on $U_p$)}\;.
 \label{funo}
 \end{equation}
\item
There exist $\zeta^A\in\Cinfty(U_p)$, verifying that
$\Lie(Y)\theta^A=\d\zeta^A$, on $U_p$; and then
$$
{\cal F}^A=\inn(Y)\theta^A-\zeta^A, \qquad \mbox{\rm (up to a
constant function, on $U_p$)}\,.
$$
 \label{fdos}
\end{enumerate}
 \label{structure}
\end{prop}
\begin{proof}
\begin{enumerate}
\item
It is a consequence of the Poincar\'e Lemma and the condition
$$
0=\Lie(Y)\omega^A=\inn(Y)\d\omega^A+\d\inn(Y)\omega^A=\d\inn(Y)\omega^A\;.
$$
\item
We have that
$$
\d\Lie(Y)\theta^A=\Lie(Y)\d\theta^A=-\Lie(Y)\omega^A=0
$$
and hence $\Lie(Y)\theta^A$ are closed forms. Therefore, by the
Poincar\'e Lemma, there exist $\zeta^A\in\Cinfty(U_p)$, verifying
that $\Lie(Y)\theta^A=\d\zeta^A$, on $U_p$. Furthermore, as
(\ref{funo}) holds on $U_p$, we obtain that
$$
\d\zeta^A=\Lie(Y)\theta^A= \d\inn(Y)\theta^A+\inn(Y)\d\theta^A=
\d\inn(Y)\theta^A-\inn(Y)\omega^A=\d \{\inn(Y)\theta^A-{\cal F}^A\}
$$
and thus \ref{fdos} holds.
\end{enumerate}
\end{proof}

{\bf Remark}:
For exact infinitesimal $k$-cosymplectic Noether
symmetries we have that
${\cal F}^A=\inn(Y)\theta^A$ (up to a constant function).

Finally, the classical Noether's theorem can be stated
for these kinds of symmetries as follows:

\begin{theorem}
\label{Nthsec} {\rm (Noether's theorem):} If $Y\in\vf(\rk\times
M)$ is an infinitesimal $k$-cosymplectic Noether symmetry then,
for every $p\in \rk \times M$, there is an open neighborhood
$U_p\ni p$ such that the functions ${\cal
F}^A=\inn(Y)\theta^A-\zeta^A$, define a conservation law  ${\cal
F}=({\cal F}^1,\ldots,{\cal F}^k)$.
\end{theorem}
\begin{proof}
Let ${\bf X}=(X_1,\ldots, X_k)\in \vf^k_H(\rk\times M)$
an integrable $k$-vector field.
From (\ref{funo}) one obtains
\beann
\sum_{A=1}^k L(X_A) {\cal F}^A &=&
 \sum_{A=1}^k \inn (X_A) d{\cal F}^A =
\sum_{A=1}^k \inn (X_A) \inn (Y)\omega^A =
 - \inn(Y)\sum_{A=1}^k \inn (X_A)\omega^A
\\ &=&   - \inn (Y)\d H +\sum_{A=1}^k\inn(Y)((\Lie(R_A)H)\eta^A)=
 - \Lie(Y)H +\sum_{A=1}^k(\Lie(R_A)H)\inn(Y)\eta^A=0 \ ,
\eeann
 that is, ${\cal F}=({\cal F}1,\ldots,{\cal F}^k)$ is a
conservation law for the Hamilton-de Donder-Weyl equations.
\end{proof}

Observe that, using Darboux coordinates in $\rk\times M$,
 the item 2 of Proposition \ref{structure}
tells us that the conservation laws
associated with  infinitesimal $k$-cosymplectic Noether symmetries
does not depend on the coordinates $(t^A)$
(as it is obvious since the generators of these symmetries,
the vector fields $Y$, neither depend on them).

\section{Example}
\protect\label{ex}

\subsection{$k$-cosymplectic quadratic Hamiltonian systems}

 Many Hamiltonian systems in field theories are of ``quadratic'' type
 and they can be modeled as follows.

 Consider the $k$-cosymplectic manifold
$(\rk\times (T^1_k)^*Q,\eta^A,\omega^A,{\cal V})$. Let $g_1,
\ldots , g_k$ be $k$ semi-Riemannian metrics in $Q$. For every
$q\in Q$ we have the following isomorphisms:
$$
\begin{array}{cccccc}
g_A^\flat & \colon & T_qQ & \longrightarrow & T_q^*Q \\
& & v & \mapsto & \inn (v)g_A
\end{array} \quad ,
$$
with $A \in \{1, \ldots , k\}$ and then we can introduce the dual
metric of $g_A$, denoted by $g_A^*$, which is defined by
$$
g_A^*(\alpha_q,\beta_q):=g_A((g_A^\flat)^{-1}(\alpha_q),(g_A^\flat)^{-1}(\beta_q))
 \quad , \quad \mbox{\rm for every $\alpha_q,\beta_q\in T_q^*Q$ and $A \in \{1, \ldots , k\}$} \quad .
 $$
We can define a function $K\in\Cinfty(\rk\times(T^1_k)^*Q)$ as follows:
for every $(t,q;\alpha_q^1,\ldots,\alpha_q^k)\in \rk\times(T^1_k)^*Q$,
$$
K(t,q;\alpha_q^1,\ldots,\alpha_q^k):= \frac{1}{2}\sum_{A=1}^k
g_A^*(\alpha_q^A,\alpha_q^A).
$$
Then, if $V \in C^{\infty}(\rk \times Q)$ we can introduce a
Hamiltonian function $H\in\Cinfty(\rk\times (T^1_k)^*Q)$ of {\sl
quadratic} type as follows
$$
H=K+ V \circ (\pi_Q)_{1,0}^*.
$$
Using natural coordinates $(t^A, q^i, p^A_i)$ on $\rk \times
(T^1_k)^*Q$ the local expression of $H$ is
 $$
 H(t^A, q^i, p^A_i)=\frac{1}{2}\sum_{A=1}^kg_{A}^{ij}(q^k)p^A_i p_j^A+V(t^B,q^j) \ ,
 $$
  where $g_{A}^{ij}$ denote the coefficients of the matrix
 associated to $g_A^*$. Then
$$
\d H=\sum_{A=1}^k\left[\frac{\partial V}{\partial t^A} \d t^A + \left(\displaystyle
\frac{1}{2} \frac{\partial g_A^{ij}}{\partial q^k} p_i^A p^A_j +
\frac{\partial V}{\partial q^k}\right)\d q^k + (g_A^{ij} p^A_i)\d p_j^A\right]
 $$
 Moreover, if ${\bf
X}=(X_1,\ldots,X_k)\in{\mathfrak{X}}^k_H(\Real^k\times(T^1_k)^*Q)$
with
 $$
  X_A=\sum_{B=1}^k\left[(X_A)^B\frac{\partial}{\ds\partial t^B} +
(X_A)^i\frac{\partial}{\ds\partial q^i}+(X_A)^B_i\frac{\partial}{\ds\partial p^B_i}\right]
$$
the equations (\ref{geonah}) lead to
\begin{equation}
\label{ec00} (X_A)^B=\delta_A^B \quad ,\quad (X_A)^i =
g_{A}^{ij}p^A_j \ \mbox{\rm ($A$ fixed)}  \quad , \quad -\displaystyle \sum_{A=1}^k
(X_A)^A_i = \frac{1}{2} \sum_{A=1}^k \frac{\partial
g_{A}^{jk}}{\partial q^i} p_{j}^A p^A_k + \frac{\partial
V}{\partial q^i}
\end{equation}
that is, we have obtained
$$
  X_A= \left[\derpar{}{t^A} +g_{A}^{ij}p^A_j\derpar{}{ q^i}+
  (X_A)^B_i\derpar{}{p^B_i}\right]
$$
with $(X_B)^B_i = -\ds\derpar{V}{q^i} - \frac{1}{2} \frac{\partial
g_A^{jk}}{\partial q^i} p^A_j p^A_k$.

Now, if $\psi(t)=(t^A,\psi^i(t),\psi^A_i(t))$ is an integral
section of ${\bf X}$ then \beq \label{ec0}
X_A(\psi(t))=\psi_*(t)\left(\derpar{}{t^A}\Big\vert_{t}\right)=
\left[\derpar{}{t^A}+\derpar{\psi^i}{t^A} \derpar{}{q^i}+
\derpar{\psi^B_i}{t^A}\derpar{}{p^B_i}\right] \eeq Thus, from
(\ref{ec00}) and (\ref{ec0}), we obtain the Hamilton-de
Donder-Weyl equations \bea -\derpar{V}{q^i}(\psi(t)) - \frac{1}{2}
\frac{\partial g_A^{ij}}{\partial q^i} \psi^A_j \psi_k^A &=&
\sum_{A=1}^k(X_A)^A_i(\psi(t)) =
\sum_{A=1}^k\derpar{\psi^A_i}{t^A}
\nonumber \\
g_{A}^{ij}(\psi(t))\psi^A_j &=& X^A_i(\psi(t))  =
\derpar{\psi^i}{t^A} \quad \mbox{\rm ($A,i$ fixed)} \ .
\label{soluc0} \eea Then, from these equations we conclude that
$$
\psi^A_i=(g_{A})_{ij}\derpar{\psi^j}{t^A} \quad \mbox{\rm ($A,i$
fixed)}   \ ,
$$
and hence the equations for the integral sections are \beq
\displaystyle \sum_{A,j}(g_A)_{ij}
\frac{\partial^2\psi^j}{\partial (t^A)^2} = -\frac{\partial
V}{\partial q^i} - \frac{1}{2} \sum_{A,j,k,l,m} \frac{\partial
g_{A}^{jk}}{\partial q^i} (g_{A})_{kl}(g_{A})_{jm} \frac{\partial
\psi^l}{\partial t^A} \frac{\partial \psi^m}{\partial t^A}, \; \;
\mbox{ for all } i . \label{hdwquad} \eeq

We also may prove the following result
\begin{prop}
Let $X$ be a Killing vector field on $Q$ for the semi-Riemannian
metrics $g_1, \ldots , g_{k}$ (that is, $\Lie(X)g_A = 0$, for all
$A \in \{1, \ldots , k\}$) such that $X(V) = 0$. Then, the vector
field $X^{1*}$ on $\rk \times (T^1_k)^*Q$ is a natural
infinitesimal symmetry for the $k$-cosymplectic Hamiltonian system
$(\rk \times (T^1_k)^*Q, H)$. Thus, if ${\cal F} = (\hat{X},
\ldots , \hat{X}): \rk \times (T^1_k)^*Q \to \rk$ is the map
defined by
\[
{\cal F}(t, q; \alpha^1_q, \ldots , \alpha^k_q) =
(\alpha^1_q(X(q)), \ldots , \alpha^k_q(X(q))),
\]
for $(t, q; \alpha^1_q, \ldots , \alpha^k_q) \in \rk \times
(T^1_k)^*Q$, we have that ${\cal F}$ is a conservation law for the
Hamiltonian system.
\end{prop}
\begin{proof}
As we know
\[
\Lie(X^{1*})\theta^A = 0, \; \; \mbox{ for all } A.
\]
Moreover, it is clear that
\[
\inn(X^{1*})\eta^A = 0, \; \; \mbox{ for all } A.
\]
So, it is sufficient to prove that
\[
\Lie(X^{1*})H = 0.
\]
Now, using that $X^{1*}$ is $(\pi_{Q})_{1,0}^*$-projectable over
$X$ and the fact that $\Lie(X)V = 0$, we deduce that
\[
\Lie(X^{1*})(V \circ (\pi_{Q})_{1,0}^*) = 0.
\]
Next, we will prove that
\[
\Lie(X^{1*})(K) = 0.
\]
Assuming that the local expression of $X$ is
\[
X = X^i \frac{\partial}{\partial q^i} \ ,
\]
then, as $\Lie(X)g_A = 0$, we have that
\[
X((g_{A})_{jk}) = \displaystyle - \frac{\partial X^l}{\partial
q^j} (g_{A})_{kl} - \frac{\partial X^l}{\partial q^k}
(g_{A})_{jl}, \; \; \mbox{ for all } A, j \mbox{ and } k
\]
which implies that
\[
X(g_{A}^{ij}) = \displaystyle - \frac{\partial X^j}{\partial q^k}
g_{A}^{ik} - \frac{\partial X^i}{\partial q^k} g_{A}^{jk}, \; \;
\mbox{ for all } A, i \mbox{ and } j.
\]
Therefore, using that the local expressions of $X^{1*}$ and $K$
are
\[
X^{1*} = \displaystyle X^i \frac{\partial}{\partial q^i} - p^A_j
\frac{\partial X^j}{\partial q^i} \frac{\partial}{\partial p^A_i},
\; \; K = \frac{1}{2} \sum_{A,i,j}g_{A}^{ij} p^A_i p^A_j
\]
we conclude that
\[
\Lie(X^{1*})K = 0.
\]
Furthermore, if $\hat{X}: T^*Q \to \mathbb{R}$ is the linear
function on $T^*Q$ associated with the vector field $X$, it
follows that
\[
(i(X^{1*})\theta^A)(t, q; \alpha^1_q, \ldots , \alpha^k_q) =
\hat{X}(\alpha^A_q), \; \; \mbox{ for all } A.
\]
Consequently, ${\cal F} = (\hat{X}, \ldots , \hat{X})$ is a
conservation law (see the remark after Proposition
\ref{structure})
\end{proof}

\subsection{A particular case: the wave equation}

As particular examples of these kinds of systems we can detache the
following case
(see \cite{MSV-2009} for a more detailed explanation):

Consider the three-dimensional wave equation
\begin{equation}\label{ec42110}
\sigma\frac{\partial^2 \psi}{\displaystyle\partial t^2}-\tau\left(\frac{\partial^2\psi}{\displaystyle\partial x^2}+
\frac{\partial^2\psi}{\displaystyle\partial y^2}+\frac{\partial^2\psi}{\displaystyle\partial z^2}\right)=0 \ .
\end {equation}
In this case $M=\Real^4\times(T^1_2)^*Q$ (i.e., $k=4$), with
$Q=\Real$ ($n=1$), and $g_i$, $i= 1, \ldots , 4$, are the
semi-Riemannian metrics on $\mathbb{R}$
\[
g_1 = \sigma dq^2, \; \; g_2 = g_3 = g_4 = -\tau dq^2,
\]
$q$ being the standard coordinate on $\mathbb{R}$. We have done
the identifications $t^1\equiv t$ and $t^2\equiv x, t^3\equiv y,
t^4\equiv z$, where $t$ is time and $x,y,z$ denote the position in
space. Then, $\psi(t,x,y,z)$ denotes the displacement of each
point of the media where the wave is propagating, as function of
the time and the position, and $\sigma$ and $\tau$ are physical
constants.

The wave equation (\ref{ec42110}) is then a particular case of the
equation (\ref{hdwquad}) for the quadratic Hamiltonian
$$
 H=\frac{1}{2}\left[\frac{1}{\sigma}(p^1)^2-
 \frac{1}{\tau}\left((p^2)^2+(p^3)^2+(p^4)^2\right)\right]
 \in\Cinfty(\Real^2\times(T^1_2)^*\Real) \ .
$$
We have that the canonical vector field on $\mathbb{R}$,
$\displaystyle \frac{\partial}{\partial q}$, is a Killing vector
field for the semi-Riemannian metrics $g_i$, $i = 1, \ldots , 4$.
Thus,
\[
{\cal F} = (p^1, p^2, p^3, p^4): \mathbb{R}^4 \times
(T^1_4)^*\mathbb{R} \to \mathbb{R}^4
\]
is a conservation law for the three-dimensional wave equation.

Note that if
\[
\tilde{\psi}: (t,x,y,z) \to (t,x,y,z,\psi(t,x,y,z);
\psi^1(t,x,y,z),\psi^2(t,x,y,z),\psi^3(t,x,y,z),\psi^4(t,x,y,z))
\]
is a solution to the Hamilton-de Donder-Weyl equations then, from
(\ref{hdwquad}), it follows that
\[
\psi^1 = \displaystyle \sigma \frac{\partial \psi}{\partial t}, \;
\; \psi^2 = -\displaystyle \tau \frac{\partial \psi}{\partial x},
\; \; \psi^3 = -\displaystyle \tau \frac{\partial \psi}{\partial
y}, \; \; \psi^4 = -\displaystyle \tau \frac{\partial
\psi}{\partial z}.
\]
Thus, the conservation law leads to the starting field equations.
In fact,
\[
Div ({\cal F} \circ \tilde{\psi}) = \sigma\frac{\partial^2
\psi}{\displaystyle\partial
t^2}-\tau\left(\frac{\partial^2\psi}{\displaystyle\partial x^2}+
\frac{\partial^2\psi}{\displaystyle\partial
y^2}+\frac{\partial^2\psi}{\displaystyle\partial z^2}\right)=0.
\]

\section{Conclusions and outlook}

We have studied symmetries and reduction of $k$-cosymplectic
Hamiltonian systems in classical field theories; in particular,
those which are modeled on $k$-cosymplectic manifolds
${\cal M}=\rk\times M$, with $M$ a generic $k$-symplectic manifold
(which we have called almost-standard $k$-cosymplectic manifolds).

In particular we have analyzed a kind of
$k$-cosymplectic Noether symmetries
for which there is a direct way to associate conservation laws
by means of the application of the corresponding generalized
version of the Noether theorem.

As discussed in Section \ref{ntmvf},
for the almost-standard $k$-cosymplectic Hamiltonian systems,
 the symmetries that we have considered in this work
have the following geometric characteristic:
they generate transformations along the fibres of the projection
 $\rk\times M\longrightarrow\rk$.
 As a consequence, in a local description, the associated conservation laws
 do not depend on the base coordinates $(t^A)$.
 This could seem to be a strong restriction but,
 really, many symmetries of field theories in physics are of this type.
 In any case, a theory of symmetries, conservation laws and reduction concerning to
 more general kinds of symmetries would have to be developed.

\section*{Acknowledgments}
We acknowledge the partial financial support of the {\sl
Ministerio de Ciencia e Innovaci\'on} (Spain), projects
MTM2008-00689, MTM2009-13383 and MTM2009-08166-E.
We thanks to the referee for his valuable suggestions and comments
which have allow us to significantly improve some parts of this work.

\end{document}